\documentclass{article}
\PassOptionsToPackage{nonatbib}{neurips_2026}
\usepackage[preprint]{neurips_2026}
\usepackage[utf8]{inputenc}
\usepackage[T1]{fontenc}
\usepackage{microtype}
\usepackage{booktabs}
\usepackage{graphicx}
\usepackage{xcolor}
\usepackage{amsmath,amssymb,amsthm}
\usepackage{enumitem}
\usepackage[backend=biber,style=numeric,sorting=none,maxbibnames=8]{biblatex}
\usepackage[hidelinks,pdfauthor={Dipankar Sarkar},pdftitle={When Privacy Moves ML-Mediated Decisions On Device: Information and Incentive Misalignment in Auctions}]{hyperref}

\newtheorem{theorem}{Theorem}
\newtheorem{proposition}[theorem]{Proposition}

\title{When Privacy Moves ML-Mediated Decisions On Device:\\
Information and Incentive Misalignment in Auctions}
\author{%
  Dipankar Sarkar\\
  Skelf Research\\
  \texttt{dipankar@skelfresearch.com}\\
  \url{https://skelfresearch.com}
}

\begin{document}
\maketitle

\begin{abstract}
Moving ML-mediated decision making onto privacy-preserving clients
decentralises the economic decision along with the inference.  Shared budget
constraints then depend on information that cannot be globally current,
creating an information-structure failure that conventional pacing is not
designed to solve.  We study this information misalignment in an
auction-logic-faithful on-device simulation with 36 campaigns and 50 devices.
Accounting is in dimensionless integer score units; no currency semantics are
claimed.  Across 30 paired demand paths, proportional Even pacing overspends
17.77\% after one tick of staleness and 1,669.31\% after 50 ticks under the
original 20-times budget pressure.  The effect does not depend on that severe a
budget: at two-times pressure, 50-tick overspend remains 106.95\%.  A
visible-budget no-sale guard makes zero-lag compliance exact at this score-unit
granularity, yet leaves 11.88\% overspend at one tick because other devices'
debits remain invisible.  A declared bursty, heterogeneous-device sweep retains
a strictly increasing mean lag curve.  We derive a finite-window expected
excess-debit bound under conditional charge caps and find positive paired slack
in every bounded-value cell.  A second,
incentive misalignment arises when the ML/pacing score transformation is
allowed to change payment units: 98.23\% of rival auctions at one tick admit a
profitable deviation.  An executable implementation-level counterexample
isolates the runner-up's multiplier in the winner's price.  Critical-base-bid
payment is per-auction DSIC conditional on current multipliers, but does not
establish dynamic truthfulness and does not repair base-value ranking
disagreement.
\end{abstract}

\section{Introduction}
\label{sec:introduction}

Moving ML-mediated decision making onto privacy-preserving clients
decentralises the economic decision along with the inference.  A local
classifier or relevance model can determine which offer is eligible and which
bid wins without disclosing its input; shared budget constraints, however,
depend on debits made across all clients.  Those debits cannot be globally
current without communication.  Economic constraints therefore become
distributed information constraints, an information-structure failure that
conventional pacing algorithms, which optimise against the state they observe,
are not designed to solve.  The accounting failure itself is not unique to
machine learning; the ML-specific motivation is that privacy-preserving
inference is what moves economically consequential eligibility and ranking
decisions onto clients while the shared monetary state stays non-local.

This information pattern is concrete rather than hypothetical.
Privacy-preserving advertising systems have explored browser-executed auctions
in isolated worklets, with buyer code querying remotely supplied signals such as
remaining campaign budget; Google's Protected Audience documentation
illustrates this information pattern~\cite{googleprotectedauction}.\footnote{We use this as a
motivating instance of the information pattern, not a claim about any API's
current deployment status; Protected Audience is itself now scheduled for
removal from Chrome~\cite{privacysandboxstatus2026}.}  A remotely supplied signal can be stale by the time a
client acts.  In our model a central service
broadcasts a pacing multiplier and spend snapshot every \(\Delta\) ticks;
between broadcasts, each of \(N\) clients combines that snapshot with only its
own recent debits.  Pacing can smooth demand against the snapshot, but cannot
repair information that has not arrived.

We study two economic failures in one auction-logic-faithful setting.
\emph{Information misalignment} occurs when local AI-mediated decisions act
against mutually inconsistent accounting state, so individually admissible
debits violate a shared budget.  \emph{Incentive misalignment} occurs when the
relevance/pacing transformation used to rank offers is also allowed to change
the units of payment: the auction charges the runner-up's scaled score rather
than the winner's critical base bid.  The first is a property of the
information structure; the second is a payment-rule defect.  Neither is cured
by a more sophisticated pacing controller.

\paragraph{Contributions.}
\begin{enumerate}[leftmargin=*,label=\arabic*.,itemsep=0.2em,topsep=0.3em]
  \item \textbf{A formal model of accounting information misalignment.}  We
  model budget state replicated across autonomous clients that observe their own
  debits but not one another's, derive a finite-window expected excess-debit
  bound from lag, arrivals, crossing campaigns, and conditional payment caps,
  and evaluate it with paired confidence intervals.
  \item \textbf{Pressure and enforceability evidence.}  A new
  \(4\times7\) pressure--lag sweep shows that long-lag excess debit persists at
  two-times as well as 20-times pressure.  A visible-balance no-sale guard
  eliminates the preventable zero-lag crossing debit but not stale
  cross-device spending.
  \item \textbf{A separate incentive-misalignment result.}  We characterise a
  canonical dominant report on rival auctions and give an executable
  implementation-level counterexample.  Critical-base-bid payment makes the
  per-auction rule DSIC conditional on current multipliers; it neither proves
  dynamic truthfulness nor changes paced ranking.
  \item \textbf{An auction-logic-faithful economic audit.}  Thirty paired
  paths use frozen budgets, disjoint calibration seeds, deposited replicate
  rows, and a live TypeScript parity test.  The audit also shows that
  score-space pacing discounts price by about 11.5 times and that PI control
  helps only at the two shortest deterministic lags.
\end{enumerate}

\section{Prior work}
\label{sec:prior-work}

\paragraph{Industrial budget pacing.}
Production pacing systems forecast supply, adapt bids, or throttle
participation to smooth campaign delivery~\cite{agarwal2014linkedin,
lee2013smooth,xu2015smartpacing,balseiro2024fieldguide}.  Agarwal et al.\
already note that delayed spend information can make fast-spending campaigns
exceed budget~\cite{agarwal2014linkedin}; staleness-causes-overdelivery is not
our discovery.  Our narrower novelty is to formalise the case in which the
  budget state is replicated across autonomous devices that see their own
debits but not one another's, derive a finite-window bound, and quantify the
  failure under the product's scoring and pacing logic.  Better feedback control
cannot undo debits that remain invisible until a sync.

\paragraph{Repeated auctions with budgets.}
Repeated-auction models analyse approximation, bidder learning, equilibrium,
regret, and efficiency under intertemporal budgets
~\cite{balseiro2015repeated,balseiro2019learning,gaitonde2023budgetpacing}.
These works explain strategic and dynamic consequences of budget constraints,
but typically give the auctioneer or learning process a coherent history of
allocations and spend.  Our devices instead make simultaneous local decisions
from inconsistent histories.  We therefore do not offer a new equilibrium or
regret result; we quantify an accounting error that those central histories
exclude.

\paragraph{Online budgeted allocation.}
The AdWords problem and later dual approaches formalise online allocation with
budgets and obtain competitive or regret guarantees
~\cite{mehta2007adwords,balseiro2023dualmirror}.  A central allocator in these
models observes the remaining resource before accepting the next item.  That
resource-feasibility assumption fails between device syncs: several devices
can accept against the same remaining score-unit budget.  Our expected-debit theorem can be
read as a finite-horizon cost of relaxing coherent resource visibility, not as
a replacement allocation algorithm.

\paragraph{Delayed-feedback online learning.}
Delayed-feedback learning relates observation delay to regret and constructs
reductions that wrap non-delayed algorithms~\cite{joulani2013delayed}.  Our lag
is different in kind: it delays a shared hard-state update, while purchases
made during the delay remain financially binding.  Consequently, the relevant
loss is excess debit after a stopping boundary, and its scale depends on
arrival intensity and payment support as well as delay.

\paragraph{Distributed quota and escrow.}
Escrow transactions partition a finite quantity among sites so each can update
locally without violating the global invariant~\cite{oneil1986escrow}.  Such
reservations are a natural alternative to the soft replicated budget state studied
here: they can make overspend impossible, but may strand budget at devices with
little matching demand.  Our analysis measures the failure of the non-escrow
design and makes that utilisation--consistency trade-off explicit.

\paragraph{Multiplicative pacing equilibria and auction foundations.}
Multiplicative pacing equilibria study autobidders whose multipliers enter
allocation and strategic optimisation~\cite{conitzer2022autobidding}; Vickrey
and Myerson give the classical relationship between monotone allocation,
critical payments, and truthful reporting
~\cite{vickrey1961counterspeculation,myerson1981optimal}.  Our negative result
is narrower and implementation-specific: the deployed rule also places the
runner-up's multiplier inside the winner's payment.  Applying a critical base
bid is a standard repair, not a new mechanism.  Our contribution is to expose
and measure the mismatch in the on-device rule while separating it from the
accounting-lag problem.

\section{Model and mechanism}
\label{sec:model}

There are campaigns \(c\in\mathcal C\), each with global budget \(B_c\), and
\(N\) devices acting during a finite window \([0,T]\).  Vertical-\(v\)
opportunities follow a marked Poisson process with intensity \(A_v\).  At a
broadcast, every device receives campaign \(c\)'s central debit snapshot and a
pacing rate \(r_c\in[0,1]\).  Until the next broadcast, device \(i\) uses that
snapshot plus only the debits it generated locally.  Sync lag is \(\Delta\)
ticks; \(\Delta=0\) denotes continuously shared accounting state, whereas \(\Delta=1\)
still permits a full tick of cross-device staleness.

For an admissible campaign, let \(b_c\) be its reported base bid and let
\(\rho_c\) combine relevance and quality; in the evaluated SDK
configuration \(0\le\rho_c\le0.8\).  The deployed auction ranks
\begin{equation}
  s_c=b_c\rho_cr_c.                                  \label{eq:score}
\end{equation}
With at least two admissible bidders, the TypeScript auction returns the
second-highest score \(\sigma\) plus one.  With one bidder it returns its
numeric CPM floor, 100, plus one; that floor is an output fallback, not a
no-sale reserve.  Although the product output is CPM-denominated, it does
not implement CPM-to-per-impression arithmetic.  We preserve the numeric rule
for parity, then define each returned number as a dimensionless score-unit
debit and round once to the nearest integer (ties away from zero).  Thus the
budget experiment is a synthetic accounting convention, not a reproduction of
product billing, and it makes no currency claim.

For allocation diagnostics, we record whether the winner differs from the
highest unscaled base value $v_c$ among admissible bidders.  We call this
\emph{base-value ranking disagreement}, rather than an efficiency or welfare
measure: because
\(\rho_c\) represents relevance and quality, a lower-base-value but
higher-relevance winner need not be socially inefficient.  We do not define or
estimate social welfare.

We compare three controllers.  As-soon-as-possible (ASAP) sets \(r_c=1\) until
exhaustion.  Even is a proportional controller with gain \(k_p=0.5\), targeting
the linear trajectory \(B_ct/T\).  PI adds an integral term with \(k_i=0.1\),
clamped to \([-1,1]\) against windup.  The reference implementation calls the
latter ``PID'', but it has no derivative term, so we use the accurate name PI.

Let realised campaign debit be \(x_c\).  Our primary estimand is pool
overspend
\begin{equation}
  \operatorname{OvSpd}=
  \frac{\sum_c(x_c-B_c)_+}{\sum_cB_c}.               \label{eq:overspend}
\end{equation}
This counts excess debit without cancelling it against underspend elsewhere.
For delivered value, let \(U(\Delta)\) be value delivered before each winning
campaign crosses its true global budget and \(D(\Delta)=\sum_c(x_c-B_c)_+\) the
excess debit; we report these separately, and where a single figure is
convenient, the normalised score
\(Q(\Delta)=[U(\Delta)-D(\Delta)]/[U(0)-D(0)]\) under a declared one-for-one
excess-debit penalty.  \(Q\) is not welfare or utility: it omits transfers,
publisher and platform value, displaced bidders, and opportunity cost.

\section{Theory}
\label{sec:theory}

\subsection{Expected overspend with stale accounting state}

Let \(Y\) be an accounting debit and \(\mathcal F_{u^-}\) the history just before an
arrival at time \(u\).  For each vertical, assume the conditional support cap
\begin{equation}
 \operatorname*{ess\,sup}
 (Y\mid\mathcal F_{u^-},v,\text{sale})\le\bar p_v.   \label{eq:payment-cap}
\end{equation}
This premise holds pointwise for bounded bids under the score-payment rule.  It
does not hold for an untruncated log-normal value distribution or for a
critical bid divided by a vanishing pacing multiplier.

\begin{theorem}[Finite-window lagged-sync overspend]
Assume the budget vector is fixed deterministically before arrivals.  Let
\(\tau_c\) be campaign \(c\)'s first budget-crossing time,
\(K_v=\sum_{c\in v}\mathbf1\{\tau_c\le T\}\), and
\(\tau_v=\min_{c\in v}\tau_c\), with \(\tau_v=T\) when \(K_v=0\).  Then, for
\(D=\sum_c(x_c-B_c)_+\) and \(B=\sum_cB_c\),
\begin{equation}
 \frac{\mathbb E[D]}{B}\le \frac1B\sum_v\bar p_v
 \mathbb E\!\left[K_v+A_v\min\{K_v\Delta,(T-\tau_v)_+\}\right].
 \label{eq:main-bound}
\end{equation}
\end{theorem}

\emph{Proof.}
Each exhausted campaign contributes at most one crossing debit.  After a
crossing, its stale availability can persist only to the next sync, an interval
of length at most \(\Delta\).  If \(I_v\) is the union of these adaptive
post-crossing intervals in vertical \(v\), then
\(|I_v|\le\min\{K_v\Delta,(T-\tau_v)_+\}\).  The indicator of \(I_v\) is
predictable, so the Poisson compensator gives \(A_v\mathbb E|I_v|\) expected
vertical arrivals in the union.  Applying Eq.~\eqref{eq:payment-cap} to the
crossing and subsequent debits, and summing across verticals, proves
Eq.~\eqref{eq:main-bound}.  The argument counts every arrival in each union and
does not assume that post-exhaustion winners are independent and identically
distributed.

The bound deliberately ignores devices' local self-knowledge and is therefore
device-count agnostic: $N$ affects realised crossings through the assignment
process, but does not appear separately from arrival intensity, lag, crossing
campaigns, and payment caps.  An $N$-dependent bound requires an explicit
arrival-to-device assignment model and is left to future work.

\subsection{Incentives under score-space payment}

The implementation stably sorts admissible campaigns by descending score, so
input order breaks an exact tie.  On an auction with a rival, fix the highest
rival score \(\sigma\) and write \(m_c=\rho_cr_c\).  The incentive
analysis concerns the unrounded auction output; integer rounding is applied
only when posting a synthetic accounting debit, and the harness tracks that
residual separately, so it does not enter the deviation test.  Bidder \(c\) wins when
\(b_cm_c>\sigma\), loses when \(b_cm_c<\sigma\), and a winner pays
\(\sigma+1\).

\begin{proposition}[The deployed rule is not truthful]
On rival auctions with non-negative reports and \(m_c>0\), a canonical
dominant report is
\begin{equation}
  b_c^*=\max\{0,(v_c-1)/m_c\}.                       \label{eq:dominant-report}
\end{equation}
For \(v_c>1\), truthful reporting under-bids exactly when
\(v_c(1-m_c)>1\) and over-bids when the inequality reverses.  If \(m_c=0\), no
report changes the zero score or allocation.
\end{proposition}

The proof is the standard threshold argument: winning yields utility
\(v_c-(\sigma+1)\), independent of \(c\)'s winning report, so the report should
cross the score threshold exactly when that utility is non-negative.  For a
concrete counterexample, take \(\rho_c=0.8\), \(r_c=0.5\), a rival bid of 1000,
and \(v_c=800\).  Truthful reporting scores 320 and loses to score 400;
reporting 1100 scores 440, pays 401, and obtains utility 399.  A unit test runs
this vector through the auction port.  The result does not cover the one-bidder
branch, where every report produces the same 101-unit sale.

If paced ranking is retained but the tentative winner instead pays
\begin{equation}
 p_c=\max\{f,\sigma/(\rho_cr_c)\},                   \label{eq:critical-price}
\end{equation}
with no sale below base-bid reserve \(f\) and no sale when \(\rho_cr_c=0\), the
payment is the critical base value.  Truthful bidding is then weakly dominant
by Myerson's lemma~\cite{myerson1981optimal}: the \emph{per-auction}
allocation/payment rule is DSIC conditional on the current relevance and
pacing multipliers.  This changes payment, not the weighted allocation rule in
Eq.~\eqref{eq:score}, and does not establish dynamic truthfulness in the
repeated process, where bids affect spend and spend affects later pacing rates.

\section{Experimental setup}
\label{sec:setup}

\paragraph{Design and product parity.}
The primary harness runs \(N=50\) homogeneous devices for \(T=100\) ticks, with 3,000
pool-wide Poisson arrivals per tick and
\(\Delta\in\{0,1,2,5,10,25,50\}\).  A 100-tick window is a nominal day, so one
tick is 14.4 minutes; we make no sub-second claim.  At each tick, each device
independently draws a Poisson number of arrivals with mean (3000/50); each
such arrival independently samples one context uniformly.  Thus opportunities
are generated by independent device streams, not assigned round-robin from a
fixed aggregate queue.  Each executes a Rust port of the TypeScript auction.
The parity test runs the live TypeScript source, including a fractional
27.9064-score vector so integer golden files cannot hide drift.  Pacing rates come from the
reference controller crate.

\paragraph{Contexts and label provenance.}
We sample 50,808 conversation-derived context records derived from
WildChat~\cite{zhao2024wildchat} (original dataset: ODC-BY~\cite{wildchatcard}); the derived
context-vector artifact released with this work is licensed CC~BY-NC~4.0.  All 50,808 were
labelled by one Apache-2.0 \texttt{gpt-oss:120b} teacher
~\cite{openai2025gptoss}; 0 of 50,808 are gold or human labels.  The deposited
200-row human-evaluation stimulus file contains no
rating columns or completed ratings.  Thus topic, intent, safety, and
ad-eligibility accuracy are unaudited.  Under the teacher labels, 5,245 records
(10.32\%) admit at least one campaign; the remainder are genuine no-fill draws
for this pool.

\paragraph{Fresh-session frequency-cap scope.}
Every sampled opportunity is defined as a fresh session.  The SDK's
three-impressions-per-campaign-per-session frequency cap therefore starts from
an empty map and never carries state between draws.  We make this explicit
choice because the deposited contexts contain neither session identifiers nor
session boundaries; imposing synthetic boundaries would add an unidentifiable
process.  The experiment has auction parity under fresh sessions, not parity
for repeated within-session exposures, and its device-count sensitivity does
not measure frequency-cap interactions.

\paragraph{Campaigns and value construction.}
The pool crosses 12 verticals with three bidders per vertical.  Premium, mid,
and value tiers bid respectively 1.00, 0.75, and 0.55 times a vertical catalogue
bid; after integer rounding, the 36 catalogue bids have mean 152.1667 score
units.  On
each opportunity the harness first draws an unscaled base value \(Z\), then
constructs bidder \(c\)'s auction value as
\begin{equation}
 v_c=\operatorname{round}\!\left(
 Z\frac{\text{catalogue bid}_c}{152.1667}\right).    \label{eq:value-construction}
\end{equation}
The deterministic arm fixes \(Z=100\) score units.  The proxy arm draws
\(Z\sim\operatorname{LogNormal}(\mu,0.85)\), with
\(\mu=\ln(68.89)-0.85^2/2\).

Zhang et al.~\cite{zhang2014rtbbenchmark} report campaign 1458's mean market
price as 68.89 Chinese fen per cost per mille (CPM; one thousand impressions),
where one fen is one hundredth of an RMB.  Market price is not advertiser
private value.  We use the published number only as a dimensionless numerical
anchor, mapping one fen-per-CPM numerical unit to one simulator score unit, so
68.89 is the mean of the \emph{unscaled} base draw \(Z\); this is neither a
currency conversion nor a claim about per-impression value.  Equation~\eqref{eq:value-construction}
then makes the actual eligible-bidder mixture's pre-rounding mean 76.3105
score units, verified from the observed vertical mix and catalogue tiers.  The
log-normal shape parameter is chosen, not estimated, and the arm is a
calibrated-synthetic sensitivity rather than an iPinYou replay.

\paragraph{Budgets and statistical protocol.}
Calibration and evaluation use disjoint demand paths.  For each primary
payment mechanism and value arm, we average unconstrained vertical spend over
pilot seeds 30,260,527--30,260,556, split it among the three vertical bidders,
divide by the 20-times pressure target, and freeze the resulting vector before
drawing evaluation seeds 20,260,527--20,260,556.  No primary budget floor is
used.  The deterministic and proxy score-rule budgets total 135,438 and 76,959
score units.  Every primary cell has \(n=30\) paired demand realisations.  Marginal
intervals are two-sided 95\% Student-\(t\) confidence intervals (CIs) on cell
means; contrasts use seed-paired differences.  Raw rows, all pilot calibration
rows, frozen vectors, and source/input hashes are deposited.

For the pressure sensitivity, we reuse those same frozen pilot means and change
only the divisor to $p\in\{2,5,10,20\}$, producing one ex-ante budget vector
per $p$ before evaluating the same 30 held-out seeds at every lag.

\paragraph{Bursty/heterogeneous robustness protocol.}
A separate declared process uses one shared two-state Markov multiplier per
tick: low intensity 0.25, high intensity 3.0625, transition probabilities 0.08
(low to high) and 0.22 (high to low), and a stationary mean of one.  Expected
state durations are 12.5 and 4.55 ticks.  Ten devices have arrival weight 2.5
and 40 have weight 0.625, also averaging one; conditional counts are independent
Poisson draws.  We keep the same frozen pressure-20 budget rather than
recalibrating and evaluate the same 30 held-out seeds at every lag.

\paragraph{Label-perturbation protocol.}
A fixed 1,558-label perturbation stress test rebuilds eligibility, recalibrates
on the disjoint pilot seeds, and reruns the Even curve; full construction and
results are in Appendix~\ref{app:labels}.

\section{Results}
\label{sec:results}

\subsection{Overspend rises sharply with sync lag}

Table~\ref{tab:overspend} gives the primary curves.  In the deterministic arm,
Even overspend rises from about 0.5\% with shared accounting state to about 18\% at one
tick and more than 16 times the aggregate budget at 50 ticks.  The
calibrated-synthetic landscape changes the curve's shape but not its direction,
reaching about nine times budget at 50 ticks.

\begin{table}[t]
\centering\scriptsize
\setlength{\tabcolsep}{3.6pt}
\caption{Pool overspend (\%), mean $\pm$ two-sided 95\% Student-$t$ CI
half-width.  Every cell has $n=30$ evaluation seeds
20,260,527--20,260,556, $N=50$ devices, 3,000 arrivals/tick, 100 ticks, a
mechanism-specific budget frozen from 30 disjoint pilot seeds, and fresh-session
frequency-cap state.  Even is proportional; PI is proportional-integral.}
\label{tab:overspend}
\begin{tabular}{rcccc}
\toprule
&\multicolumn{2}{c}{deterministic}&\multicolumn{2}{c}{calibrated-synthetic}\\
\cmidrule(lr){2-3}\cmidrule(lr){4-5}
$\Delta$&Even&PI&Even&PI\\
\midrule
0&$0.5276\pm0.0347$&$0.2800\pm0.0294$&$0.9983\pm0.0487$&$0.8148\pm0.0407$\\
1&$17.7731\pm0.8702$&$7.4418\pm0.7871$&$31.8700\pm1.2762$&$33.1133\pm1.0317$\\
2&$50.4252\pm1.4789$&$47.1974\pm1.2682$&$72.1369\pm2.1281$&$70.4868\pm2.2735$\\
5&$189.9383\pm1.8070$&$189.9383\pm1.8070$&$109.4431\pm12.0941$&$108.6332\pm12.6605$\\
10&$479.2772\pm2.5343$&$479.2772\pm2.5343$&$124.0767\pm4.6931$&$124.0767\pm4.6931$\\
25&$1346.1324\pm3.7170$&$1346.1324\pm3.7170$&$402.7920\pm2.4103$&$402.7920\pm2.4103$\\
50&$1669.3075\pm3.8743$&$1669.3075\pm3.8743$&$905.0317\pm4.0435$&$905.0317\pm4.0435$\\
\bottomrule
\end{tabular}
\end{table}

Device count and offered load both matter.  At \(\Delta=1\), fixed aggregate
load, and the same frozen deterministic budget, Even overspend is about 13\%,
18\%, and 20\% for \(N=10,50,200\).  Holding \(N=50\) and the budget fixed while
halving or doubling arrival load yields approximately 0\%, 18\%, and 47\%.
These are one-factor sensitivities, not a scale law.

\subsection{The lag effect persists under milder budget pressure}

Figure~\ref{fig:pressure} varies only the denominator used to freeze budgets.
The 20-times curve exactly reproduces all 210 corresponding primary raw rows.
As expected, pressure controls magnitude, and controller feedback makes the
short-lag curves mildly nonmonotone at two- and five-times pressure.  The
long-lag information failure nevertheless remains: at \(\Delta=50\), mean
overspend is 106.95\%, 367.37\%, 801.41\%, and 1,669.31\% for pressure 2, 5,
10, and 20, respectively.

\begin{figure}[t]
  \centering
  \includegraphics[width=0.93\linewidth]{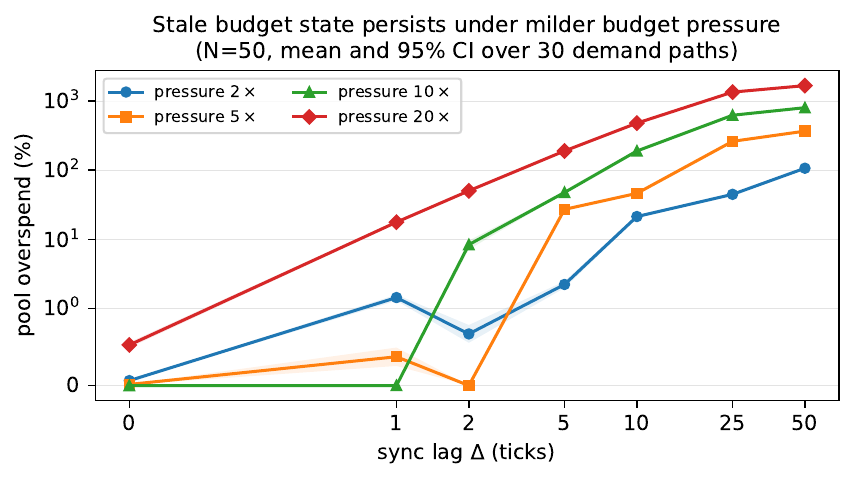}
  \caption{Budget-pressure sensitivity for deterministic Even pacing.  Each
  curve uses a budget vector frozen from the same 30 disjoint pilot paths after
  dividing mean unconstrained spend by the labelled pressure; points and bands
  are means and two-sided 95\% Student-(t) CIs over the same 30 evaluation
  seeds.  Milder budgets reduce magnitude but do not remove long-lag excess
  debit.}
  \label{fig:pressure}
\end{figure}

\subsection{The lag effect survives bursty, heterogeneous demand}

Under the declared process, mean overspend at
\(\Delta=0,1,2,5,10,25,50\) is respectively 0.4861\%, 39.5224\%, 118.8324\%,
241.8145\%, 381.9290\%, 1,098.1294\%, and 1,573.1751\%.  Thus the mean curve is
strictly increasing; every paired lag-minus-zero 95\% CI is positive, and 24
of 30 individual paths are themselves strictly increasing across all seven
cells.  This sweep supports the direction of the lag effect under one specified
burst/device mixture, not its deployment magnitude or robustness to arbitrary
demand processes.

\subsection{Visible-balance rejection removes only local crossing overshoot}

The guarded ablation rejects a tentative sale when its rounded score-unit debit
exceeds the winning campaign's remaining device-visible budget.  At zero
lag it produces exactly 0 excess debit in all 30 replicates, eliminating the
primary arm's $0.5276\pm0.0347$\% crossing overshoot.  At one tick, however,
it still overspends $11.8822\pm1.1969$\%: each device respects what it can see
while other devices' unsynchronised debits remain absent.  Guard and standard
means are equal only at \(\Delta=5\) (both 189.9383\%) and \(\Delta=10\) (both
479.2772\%); at 25 ticks they are 1,345.6427\% versus 1,346.1324\%, and at 50
ticks 1,668.6929\% versus 1,669.3075\%.  The guard and zero-lag crossing result
are artefacts of the chosen integer score granularity and are not claimed to
transfer to a real currency ledger.

\subsection{The expected bound contains the measured curve}

The 12 corpus-derived vertical intensities sum to about 309.7 servable arrivals
per tick.  The charge cap $\bar p_v$ in Eq.~\eqref{eq:main-bound} is an
\emph{ex-ante} bound on the conditional support of a sale's charge, not a
realised path maximum: on the bounded deterministic arm
$\bar p_v=\max(101,\,\mathrm{round}(0.8v_{(2)}+1))$ is fixed before any run by
the vertical's second-highest admissible value $v_{(2)}$, the SDK quality
ceiling $0.8$ and the one-unit increment, with the single-bidder case falling
back to the floor plus one.  These bounds are $101$--$119$ score units;  Table~\ref{tab:bound}
evaluates Eq.~\eqref{eq:main-bound} with paired estimators.  Every slack CI is
positive.  This is consistency of estimated expectations, not a claim that the
bound must hold on each finite replicate.  The untruncated proxy arm is outside
the theorem's support premise.

\begin{table}[t]
\centering\small
\caption{Expected-bound evaluation for deterministic Even pacing.  Empirical
and bound columns are means over the same $n=30$ evaluation seeds; ``paired
slack'' is bound minus empirical overspend in percentage points with a
two-sided 95\% Student-$t$ CI.  Budgets are frozen from 30 disjoint pilot seeds;
$N=50$, 3,000 arrivals/tick, 100 ticks, and fresh sessions in every row.}
\label{tab:bound}
\begin{tabular}{rrrr}
\toprule
$\Delta$&empirical (\%)&bound (\%)&paired slack [95\% CI]\\
\midrule
1&17.7731&70.4754&52.7024 [51.3435,54.0613]\\
2&50.4252&145.6363&95.2111 [93.7322,96.6900]\\
5&189.9383&359.9642&170.0258 [168.2188,171.8328]\\
10&479.2772&717.1772&237.9001 [235.3658,240.4344]\\
25&1346.1324&1788.8165&442.6841 [438.9671,446.4011]\\
50&1669.3075&2357.8528&688.5453 [684.7028,692.3878]\\
\bottomrule
\end{tabular}
\end{table}

\subsection{Payment incentives and ranking remain distinct}

At deterministic \(\Delta=1\), 98.23\% of rival auctions admit a profitable
deviation under the score-space rule; including the outcome-invariant
one-bidder branch, the rate is 96.97\% of all sales.  The highest-value
admissible bidder differs from the winner in 51.63\% of sales.  Under
critical-base-bid payment with a 100-unit no-sale reserve, the harness finds no
profitable per-auction deviation in any of 420 replicate cells, conditional on
the realised multipliers.  Paced ranking remains: zero-lag deterministic
base-value ranking disagreement is 49.26\%.  This is not labelled an
inefficiency because relevance may itself be part of the allocation objective.

The payment change also affects delivery.  At the deliberately smaller
20,000-opportunity factorial scale, we cross score versus critical payment,
0 versus 100-unit base-bid reserve, and an unclamped versus 500-unit-floored
budget vector.  The unclamped and floored aggregate budgets are respectively
9,000 and 20,265 score units and remain identical across payment, reserve, and seeds.
Without a reserve and with the unclamped vector, critical payment adds about
11.45 underspend points (95\% CI [9.91,12.99]).  For critical payment, adding
the reserve adds about 33.56 points [31.53,35.58]; conditional on that reserve,
the floor adds about 12.32 points [11.52,13.13].  Reserve and floor effects are
therefore larger than the isolated payment effect.

\subsection{PI has a narrow advantage}

Even and PI share each evaluation seed, so the controller estimand is the
replicate-level difference PI minus Even.  It is detectably negative only at
deterministic \(\Delta=1\), where the effect is about \(-10.33\) points (95\%
CI [\(-11.46\),\(-9.21\)]), and \(\Delta=2\), where it is about \(-3.23\)
[\(-4.12\),\(-2.34\)].  All calibrated-synthetic effects are inconclusive.
The two controllers are exactly equal from deterministic \(\Delta=5\) and
proxy \(\Delta=10\) onward because their paced states converge to the same
boundary behaviour.  The evidence supports a short-lag controller advantage,
not a general PI remedy for stale accounting.

\subsection{Label perturbation leaves the lag result intact}

A 1,558-label perturbation stress test leaves the monotone lag effect
unchanged: every paired lag-minus-zero CI remains positive and all 30
individual perturbed curves are strictly increasing.  Construction and cell
results appear in Appendix~\ref{app:labels}; this is a stress test, not a human
audit.

\subsection{Pacing discounts price rather than participation}

At deterministic zero lag, Even wins 19,358.1 impressions at an average numeric
charge of 7.07 score units, while ASAP wins 1,679.0 at 81.35.  Thus pacing multiplies delivery
by about 11.5 and divides price by the same factor.  The paced pool buys about 62.5\%
of expected servable inventory.  This behaviour follows from placing the
pacing rate in the runner-up score that determines payment: lower multipliers
lower prices even when the auction still selects a winner.  It is qualitatively
different from probabilistic participation throttling.

Reporting authorised delivered value separately from excess debit matters: the
tempting affine shortcut \(Q=1-\operatorname{OvSpd}\) holds in only 67 of the
1,800 raw primary and ablation rows, all zero-overspend coincidences from the
explicitly labelled probabilistic-throttling ablation, and in none of the 1,380
rows that execute the auction mechanism.  The delivered-value quantity is
therefore not a rescaling of overspend.

\section{Analysis and limitations}
\label{sec:analysis}

\paragraph{The economic bottleneck is consistency, not controller tuning.}
The steep lag curve and narrow PI improvement point to the same diagnosis.  A
controller acts on the state it can see; after a campaign crosses budget,
dozens of devices can continue authorising score-unit debits until the state converges.
Faster sync reduces the contaminated interval, while hard device reservations
would remove excess debit at the cost of stranded allocations.  Choosing
between broadcast frequency and escrow-like reservation is therefore an
economic trade-off between communication, overspend, and utilisation.

\paragraph{Payment and ranking require separate analysis.}
Conditional on current multipliers, critical-base-bid payment fixes the
per-auction incentive defect of charging in score space, but it neither makes
the accounting state coherent nor removes paced ranking.
Conversely, perfect state sync would stop most excess debit while leaving the
score-payment deviation.  A deployment should not treat one mechanism change
as a cure for budget consistency, conditional per-auction truthfulness, and
the platform's chosen ranking objective; these are distinct design axes.

\paragraph{Unaudited labels and synthetic demand.}
The context pool is entirely single-teacher labelled, with no human-gold rows
or completed human ratings.  The perturbation analysis brackets one material
error construction, but cannot estimate real precision, recall, taxonomy
confusion, or safety error.  The 12-vertical, three-bidder pool is synthetic and
sparser than a production market.  These limitations restrict the numerical
curve, even though the replicated-state mechanism does not depend on the
teacher's semantic correctness.

\paragraph{Session, arrival, and scale assumptions.}
Each draw is a fresh session, so the product's within-session three-impression
cap never binds.  Results cannot be extrapolated to repeated-session inventory
without session identifiers and a cap-state model.  The primary curve uses
homogeneous independent Poisson streams and independently sampled contexts.
The added robustness arm covers one declared Markov-modulated burst process and
one fixed high/low device mixture, but not diurnal cycles, trace-linked device
contexts, network failures, or strategic timing.  Both designs fix \(N=50\),
100 ticks, and a baseline 3,000 arrivals per tick.

\paragraph{Value, bidding, and theorem scope.}
All delivery runs assume truthful reports even though the deployed payment
rule makes truth-telling suboptimal.  We do not solve the resulting repeated
strategic game: bids affect current spend, which changes later pacing rates.
The deterministic and log-normal value arms are proxy
landscapes, and the latter imports only one iPinYou market-price mean with an
author-chosen tail shape.  Because the log-normal is untruncated, it has no
finite conditional charge cap and lies outside Theorem~1.  All simulator values,
budgets, and debits are dimensionless score units despite legacy internal field
names.  Integer rounding is a chosen accounting granularity: the visible-budget
guard and small-lag crossing residuals are unit-sensitive, and we withdraw any
claim that they transfer to a real currency ledger.

\paragraph{Broader impacts.}
Auditing stale shared constraints can clarify accountability in
privacy-preserving ML systems, but this score-unit study does not quantify real
advertiser charges.  Conversely,
better budget utilisation can make behavioural advertising more effective;
the underlying teacher labels may encode demographic or linguistic bias, and
on-device execution alone does not make targeting fair or harmless.  We
therefore present an accounting and mechanism audit, not a recommendation to
expand personalised advertising.

\section{Conclusion}
\label{sec:conclusion}

When ML-mediated economic decisions move onto clients, a shared budget becomes
an information constraint.  Pressure and demand-robustness sweeps retain the
lag direction; the visible-budget guard isolates local crossing from debits
hidden on other devices at the chosen granularity.  Pacing cannot repair unseen state.

Charging the runner-up's paced score also makes the per-auction rule
non-truthful.  Conditional critical-base-bid payment repairs that incentive but
leaves paced ranking, dynamic strategy, and distributed accounting untouched.
The direction survives our label and demand stress tests; magnitudes remain
conditional on fresh sessions, truthful play, synthetic campaigns, unaudited
labels, and dimensionless accounting, and do not estimate a currency ledger.

\section{Future work}

Priorities are an independent human label audit, sessionised logs for the
three-impression cap, and trace-linked device demand.  Mechanism work should
compare faster sync, escrow-style device reservations and budget transfer,
measuring excess debit against stranded inventory; an explicit
arrival-to-device model could yield an $N$-dependent bound exploiting local
self-knowledge.  Finally, the repeated reporting--pacing game must be solved or
simulated before truthful-report delivery is treated as an equilibrium
prediction.

\printbibliography[title={References}]

\appendix

\section{Label-perturbation stress test}
\label{app:labels}

Because no audit can estimate semantic error, this is a constructed sensitivity
test rather than a confidence claim about the teacher.  With a fixed label
seed, we mark 10\% of the 5,245 originally servable records as false positives
(525 records), make 1\% of the full corpus false negatives using complete label
tuples from the empirical servable mix (508 records), and assign a wrong
vertical to a disjoint 10\% of originally servable records (525 records).  The
1,558 changes leave 5,228 records servable.  We rebuild eligibility, recalibrate
on the same 30 disjoint pilot seeds, freeze the new budget, and evaluate the
same 30 held-out seeds.

\begin{table}[h]
\centering\small
\caption{Label-perturbation sensitivity for deterministic Even overspend
(\%).  Entries are mean $\pm$ two-sided 95\% Student-$t$ CI half-width.}
\label{tab:labels}
\begin{tabular}{rrr}
\toprule
$\Delta$&original&perturbed\\
\midrule
0&$0.5276\pm0.0347$&$0.5367\pm0.0414$\\
1&$17.7731\pm0.8702$&$18.3370\pm1.2525$\\
2&$50.4252\pm1.4789$&$49.6566\pm1.0992$\\
5&$189.9383\pm1.8070$&$190.3281\pm1.7013$\\
10&$479.2772\pm2.5343$&$480.3457\pm2.2025$\\
25&$1346.1324\pm3.7170$&$1347.0086\pm3.7100$\\
50&$1669.3075\pm3.8743$&$1667.7778\pm3.9587$\\
\bottomrule
\end{tabular}
\end{table}

All seven paired changed-minus-original intervals include zero.  Every paired
lag-minus-zero CI is positive, and all 30 perturbed replicate curves are
strictly increasing.  At one tick the perturbed lag-minus-zero effect is 17.80
points [16.57,19.03]; at 50 ticks it is 1,667.24
[1,663.29,1,671.20].  This supports robustness to the stated construction, not
semantic validity.

\section{Reproducibility, artifacts, and compute}
\label{app:reproducibility}

The supplemental artifact is available at
\url{https://github.com/sarkar-dipankar/on-device-auction-audit}.
It contains Rust and Python producers,
evaluation context vectors, frozen budget vectors, every pilot-calibration row,
every evaluation replicate, manifests, input/source hashes, and deterministic
source snapshots.  The primary experiments contain 1,800 raw rows; the new
pressure and visible-balance artifact adds 1,050 rows.  Its 20-times rows match
all 210 corresponding primary Even rows exactly.  The bursty/heterogeneous
artifact adds 210 rows.  The auction parity suite executes the live TypeScript
implementation, including a fractional-score test vector, against the Rust port.

The simulator is run from the repository root with the stable Rust toolchain;
derived statistics and figures use the checked-in Python environment with byte
code disabled.  Manifests record exact commands, seeds, toolchain versions,
hashes, and schema descriptions.  On an AMD Ryzen 7 5700U CPU (8 physical/16
logical cores) with 62 GiB RAM, the new 1,050-row CPU-only sweep used 16
workers and took 53.5 seconds wall time; the 210-row robustness sweep took 5.8
seconds.  Each raw row also records its cell time.  The new manifests record the
producer, all behaviour-relevant shared-source hashes, repository state, exact
Rust/Cargo/Node versions, and deterministic source-snapshot hash.  Wall times
for the earlier primary, factorial, and label runs were not recorded, so total
project compute cannot be reconstructed exactly.

\section*{NeurIPS Paper Checklist}

\begin{enumerate}[leftmargin=*,itemsep=0.8em]

\item {\bf Claims}
\item[] Question: Do the main claims made in the abstract and introduction accurately reflect the paper's contributions and scope?
\item[] Answer: \answerYes{}
\item[] Justification: The abstract and Section~\ref{sec:introduction} separate the information and incentive claims; Sections~\ref{sec:results} and~\ref{sec:analysis} state their empirical scope and limitations.

\item {\bf Limitations}
\item[] Question: Does the paper discuss the limitations of the work performed by the authors?
\item[] Answer: \answerYes{}
\item[] Justification: Section~\ref{sec:analysis} discusses unaudited single-teacher labels, synthetic demand and values, fresh sessions, the primary homogeneous and declared bursty/heterogeneous arrivals, dynamic strategy, theorem scope, and broader impacts.

\item {\bf Theory assumptions and proofs}
\item[] Question: For each theoretical result, does the paper provide the full set of assumptions and a complete (and correct) proof?
\item[] Answer: \answerYes{}
\item[] Justification: Section~\ref{sec:theory} states the deterministic-budget, Poisson-arrival, and conditional-charge-cap assumptions and gives the proof of Theorem~1 and the threshold proof of Proposition~2.

\item {\bf Experimental result reproducibility}
\item[] Question: Does the paper fully disclose all the information needed to reproduce the main experimental results of the paper to the extent that it affects the main claims and/or conclusions of the paper (regardless of whether the code and data are provided or not)?
\item[] Answer: \answerYes{}
\item[] Justification: Section~\ref{sec:setup} and Appendix~\ref{app:reproducibility} specify the assignment process, mechanisms, values, budgets, disjoint seed blocks, statistical method, and artifact contents.

\item {\bf Open access to data and code}
\item[] Question: Does the paper provide open access to the data and code, with sufficient instructions to faithfully reproduce the main experimental results, as described in supplemental material?
\item[] Answer: \answerYes{}
\item[] Justification: The supplemental archive contains the Rust and Python producers, context-vector input, live TypeScript parity tests, frozen budgets, raw pilot and evaluation rows, manifests, hashes, and run instructions (Appendix~\ref{app:reproducibility}).

\item {\bf Experimental setting/details}
\item[] Question: Does the paper specify all the training and test details (e.g., data splits, hyperparameters, how they were chosen, type of optimizer) necessary to understand the results?
\item[] Answer: \answerYes{}
\item[] Justification: There is no model training in this paper.  Section~\ref{sec:setup} specifies all simulation parameters, controller gains, value distributions, calibration/evaluation separation, and the fresh-session decision.

\item {\bf Experiment statistical significance}
\item[] Question: Does the paper report error bars suitably and correctly defined or other appropriate information about the statistical significance of the experiments?
\item[] Answer: \answerYes{}
\item[] Justification: Every headline cell uses 30 demand paths; Section~\ref{sec:setup}, tables, and figure captions define two-sided 95\% Student-$t$ intervals, with seed-paired intervals for contrasts.

\item {\bf Experiments compute resources}
\item[] Question: For each experiment, does the paper provide sufficient information on the computer resources (type of compute workers, memory, time of execution) needed to reproduce the experiments?
\item[] Answer: \answerNo{}
\item[] Justification: Appendix~\ref{app:reproducibility} reports CPU, memory, worker count, per-cell timing, 53.5-second wall time for the 1,050-row sweep, and 5.8 seconds for the 210-row robustness sweep.  Earlier runs did not record wall time, so full project compute is not reconstructible.

\item {\bf Code of ethics}
\item[] Question: Does the research conducted in the paper conform, in every respect, with the NeurIPS Code of Ethics \url{https://neurips.cc/public/EthicsGuidelines}?
\item[] Answer: \answerYes{}
\item[] Justification: The work is an offline simulation over a licensed, consent-documented derivative; it makes no deployment or human-audit claim and explicitly discusses data and advertising harms.

\item {\bf Broader impacts}
\item[] Question: Does the paper discuss both potential positive societal impacts and negative societal impacts of the work performed?
\item[] Answer: \answerYes{}
\item[] Justification: The broader-impacts paragraph in Section~\ref{sec:analysis} covers reduced unintended charges and accountability as well as more effective advertising and label-bias risks.

\item {\bf Safeguards}
\item[] Question: Does the paper describe safeguards that have been put in place for responsible release of data or models that have a high risk for misuse (e.g., pre-trained language models, image generators, or scraped datasets)?
\item[] Answer: \answerNA{}
\item[] Justification: This paper releases no trained or generative model and no raw conversations; the simulation input contains derived context vectors.  The source dataset's consent, redaction, access, and license conditions remain applicable.

\item {\bf Licenses for existing assets}
\item[] Question: Are the creators or original owners of assets (e.g., code, data, models), used in the paper, properly credited and are the license and terms of use explicitly mentioned and properly respected?
\item[] Answer: \answerYes{}
\item[] Justification: Section~\ref{sec:setup} cites WildChat and states the derivative's CC BY-NC 4.0 license; it cites the Apache-2.0 \texttt{gpt-oss:120b} model card.  The supplemental cards preserve provenance and terms.

\item {\bf New assets}
\item[] Question: Are new assets introduced in the paper well documented and is the documentation provided alongside the assets?
\item[] Answer: \answerYes{}
\item[] Justification: The supplemental archive documents code, schemas, dataset/model provenance, manifests, frozen budgets, raw rows, toolchains, limitations, and licenses; Appendix~\ref{app:reproducibility} summarises these materials.

\item {\bf Crowdsourcing and research with human subjects}
\item[] Question: For crowdsourcing experiments and research with human subjects, does the paper include the full text of instructions given to participants and screenshots, if applicable, as well as details about compensation (if any)?
\item[] Answer: \answerNA{}
\item[] Justification: This paper conducts no crowdsourcing and recruits or interacts with no participants.  The deposited but uncompleted human-evaluation stimulus is not treated as a study or evidence.

\item {\bf Institutional review board (IRB) approvals or equivalent for research with human subjects}
\item[] Question: Does the paper describe potential risks incurred by study participants, whether such risks were disclosed to the subjects, and whether Institutional Review Board (IRB) approvals (or an equivalent approval/review based on the requirements of your country or institution) were obtained?
\item[] Answer: \answerNA{}
\item[] Justification: No human-subject experiment was conducted.  The work uses a previously released, consent-documented and redacted dataset derivative and does not access participant identities.

\item {\bf Declaration of LLM usage}
\item[] Question: Does the paper describe the usage of LLMs if it is an important, original, or non-standard component of the core methods in this research? Note that if the LLM is used only for writing, editing, or formatting purposes and does \emph{not} impact the core methodology, scientific rigor, or originality of the research, declaration is not required.
\item[] Answer: \answerYes{}
\item[] Justification: Section~\ref{sec:setup} identifies the single \texttt{gpt-oss:120b} teacher, states that all 50,808 labels are unaudited and none are human-gold, and Appendix~\ref{app:labels} reports the perturbation stress test without presenting it as a human audit.

\end{enumerate}

\end{document}